\documentclass[12pt]{article}
\usepackage{epsfig,multicol,bbm,amsmath,amssymb,amscd,mathrsfs,fancybox,amsthm,framed,enumerate,txfonts, booktabs}
\usepackage{fancyhdr}

\usepackage[top=20truemm,bottom=25truemm,left=20truemm,right=20truemm]{geometry}
 \makeatletter
    
    \@addtoreset{equation}{section}
  \makeatother
\newtheorem{Thm}{Theorem}[section]
\newtheorem{Lem}{Lemma}[section]

\newtheorem{Ass}{Assumption}[section]

\newcommand{\beq}{\begin{align}}
\newcommand{\eeq}{\end{align}}

\newcommand{\Expect}[1]{\left\langle{#1}\right\rangle}

\newcommand{\no}{\nonumber}

\newcommand{\Natural}{\mathbb{N}}

\newcommand{\Real}{\mathbb{R}}

\newcommand{\e}{\mathrm{e}}

\newcommand{\q}{\quad}
\renewcommand{\H}{\mathcal{H}}

\newcommand{\F}{\mathcal{F}}
\DeclareMathOperator*{\op}{\oplus}

\DeclareMathOperator*{\ot}{\otimes}

\DeclareMathOperator*{\hot}{\hat{\otimes}}
\newcommand{\R}{\mathbb{R}}
\newcommand{\C}{\mathbb{C}}
\newcommand{\kk}{\mathbf{k}}

\newcommand{\xx}{\mathbf{x}}

\newcommand{\bb}{\mathrm{b}}

\newcommand{\ee}{\mathbf{e}}

\renewcommand{\rm}{\mathrm}

\renewcommand{\hat}{\widehat}

\begin{document}
\title{New criteria for self-adjointness and its application to Dirac-Maxwell Hamiltonian}
\maketitle
\begin{center}
Shinichiro Futakuchi and Kouta Usui
\vskip 2mm
\begin{small}
\textit{Department of Mathematics, Hokkaido University}\\
 \textit{060-0810, Sapporo, Japan.}

\end{small}
\end{center}
\abstract{
We present a new theorem concerning a sufficient condition for a symmetric operator acting in
a complex Hilbert space to be essentially self-adjoint.
By applying the theorem, we prove that the Dirac Maxwell Hamiltonian,
which describes a quantum system of a Dirac particle and a radiation field minimally interacting
with each other, is essentially
self-adjoint. Our theorem covers the case where the Dirac particle is in the Coulomb
type potential.}

\vskip 5mm
\noindent\textbf{keywords} : symmetric operator, essentially self-adjoint operator, Quantum Electrodynamics  \\
\noindent \textbf{Mathematical Subject Classifilations} : 81Q10, 47N50

\section{Introduction}
One of the most important mathematical studies of quantum systems
is to prove the self-adjointness of the Hamiltonian.
A self-adjoint Hamiltonian generates a unique time evolution operator,
while symmetric but not self-adjoint Hamiltonians may generate
no natural time evolution or may generate
a lot of different dynamics, because they 
have, in general, no self-adjoint extensions or infinitely many ones.  Moreover, the 
``probability interpretation" in quantum theory crucially depends upon the existence of
a spectral measure supported on the real line, which belongs \textit{only} to self-adjoint operators. 
In these viewpoints, proving the self-adjointness of a Hamiltonian is not
 just a problem on a mathematical 
technicality but also of \textit{physical} importance, and therefore developing general 
mathematical theorems for the
self-adjointness would contribute both to mathematics and to physics.

The Dirac-Maxwell model is expected to describe a quantum system
consisting of a Dirac particle and a radiation field with the minimal interaction.
Informal perturbation method shows that this model derives
the Klein-Nishina formula for the cross section
of the Compton scattering of an electron and a photon, which agrees with
the experimental results very well \cite{Nishijima1973a}. Hence, it is strongly suggested that the Dirac-Maxwell model describes a class of natural phenomena 
where the quantized radiation field plays an essential role.
The mathematically rigorous study of this model was initiated by Arai in Ref. \cite{MR1765584},
and several mathematical aspects of the model was analyzed so far (see, e.g.,
\cite{MR1981623}, \cite{MR2260374}, \cite{MR2810826}, \cite{MR2178588}, and \cite{MR2377946}).
The Hamiltonian of this model has a certain singularity coming from the fact that the free part 
Hamiltonian is not bounded below, which is quite special among Hamiltonians of realistic quantum systems.  
The essential self-adjointness of 
the Dirac-Maxwell Hamiltonian 
has been analyzed in Refs.  \cite{MR1765584} and \cite{MR2377946},
but, to our best knowledge, the proof of the essential self-adjointness
 in the case where the Dirac particle lies in the Coulomb
type potential, is still missing, although this is one of the most important situations in physics. 
The main goal of the present paper
is to give a proof of it.
  
The present paper can be regarded as a sequel to Ref. \cite{FutakuchiUsui2013}.
In Ref. \cite{FutakuchiUsui2013}, the authors developed a general theory on the existence of
solutions of initial value problems for the Schr\"odinge and Heisenberg equations generated by
a linear operator $H$ in some Hilbert space $\H$. One of the merits of the general theory 
established there is that it is applicable to the case where $H$ is not symmetric or even not normal,  
but it will also help us to attack mathematical 
problems of quantum field theories with a usual symmetric Hamiltonian. 
In this paper, we develop a new theorem for a symmetric operator to
be essentially self-adjoint, making a good use of materials 
obtained in Ref. \cite{FutakuchiUsui2013}. Assumptions we employ here seem to be
compatible with a large class of Hamiltonians of mathematical quantum field models with
the interaction being linear in bosonic field operators. In fact, the Dirac-Maxwell
Hamiltonian with the Coulomb type potential fulfills our assumptions and turns
 out to be essentially self-adjoint on
its natural domain, in spite of the fact that it is not semi-bounded.

\section{General Theorem on essential Self-adjointness}
In this section, we develop a general strategy for proving the essential self-adjointness,
by using the results obtained in Ref. \cite{FutakuchiUsui2013}.
Let $\mathcal{H}$ be a complex Hilbert space with the inner product $\Expect{\cdot,\cdot}$.
For a linear operator $ T $ in $ \H $, we denote, in general, its domain and range by $ D(T) $ and $ R(T) $, respectively. We also denote the adjoint of $T$ by $ T^* $ and the closure by $ \overline{T} $, if these exist. For a self-adjoint operator $ T $,  $ E_T (\cdot ) $ denotes the spectral measure of $ T $.

Let $H_0$ be a self-adjoint operator in $\mathcal{H}$, and $H_1$ be a symmetric one 
in $\mathcal{H}$. Suppose that $ D(H_0) \cap D(H_1) $ is dense in $ \H $. The main object we consider in the present section is the symmetric operator
\begin{align}\label{Hdef}
H:=H_0+H_1.
\end{align}
Assumptions we employ here are as follows \cite{FutakuchiUsui2013} :
\begin{Ass}\label{ass1} There exists an operator $ A $ in $ \H $ satisfying the following conditions:
\begin{enumerate}[(I)]
\item $A$ is self-adjoint and non-negative.
\item $A$ and $H_0$ are strongly commuting.
\item $H_1$ is $A^{1/2}$- bounded, where $ A^{1/2} $ defined through operational calculus.
\item There exists a constant $b>0$ such that, for all $L\ge 0$, $\xi \in R( E_A([0,L]))$
implies $H_1 \xi \in R( E_A([0,L+b]))$.
\end{enumerate}
\end{Ass}
\noindent For the notational simplicity, we put
\begin{align}
D := \bigcup _{L \ge 0} R(E_A ([0, L])) ,
\end{align}
and $D':=D(H_0)\cap D$.
Note that $D$ is a dense subspace in $\H$, because $A$ is self-adjoint by Assumption \ref{ass1} (I),
 and $D'$ is also dense if Assumption \ref{ass1} (I) and (II) are valid \cite{FutakuchiUsui2013}.

Our goal of the present section is to prove the following theorem :

\begin{Thm}\label{abs-main}
Suppose that Assumption \ref{ass1} holds.
\begin{enumerate}[(i)]
\item \label{main1} 
Exactly one of the following (a) or (b) holds.
\begin{enumerate}[(a)] 
\item $ H $ has no self-adjoint extension.
\item $ H $ is essentially self-adjoint.
\end{enumerate}
\item \label{main2} If $D'$ is a core of $H$,
then $H$ is essentially self-adjoint on $D'$.
\end{enumerate}
\end{Thm}
To prove Theorem \ref{abs-main}, we first recall the result obtained in Ref. \cite{FutakuchiUsui2013}.
Define for $t\in \R$,
\[ H_1(t):=e^{itH_0}H_1e^{-itH_0} .\]
All we need here are summarized as the following lemma: 
\begin{Lem}\label{time-order} 
Suppose that Assumption \ref{ass1} holds. Then, for each $ t,t' \in \R $ and $ \xi \in D $, the series
\begin{align}
U(t,t') \xi := \xi + (-i) \int _{t'} ^t d\tau _1 \, H_1(\tau _1) \xi + (-i)^2 \int _{t'} ^t d\tau _1 \int _{t'} ^{\tau _1} d\tau _2 \, H_1 (\tau _1) H_1 (\tau _2) \xi + \cdots 
\end{align}
converges absolutely, where each of integrals is taken in the sense of strong integral. 
Furthermore, the following (i) and (ii) hold. 

\begin{enumerate}[(i)]
\item The operator $ U(t,t') $ has a unitary extension $ \overline{U(t,t')} $ uniquely.

\item For each $ \xi \in D' $, put $\xi (t) := e^{-itH_0}\overline{U(t,0)}\xi $. Then, $\xi(t)\in D(H)$ for all $t\in \R$,
and the vector valued function $t\mapsto \xi(t)$ is strongly differentiable. Moreover, 
it is a solution of the initial value problem for the Schr\"{o}dinger equation:
\begin{align}\label{sch}
\frac{d}{dt} \xi (t) = -iH \xi (t) , \q \xi (0) = \xi ,
\end{align}
where the derivative in $t$ is taken in the strong sense (this applies in what follows unless otherwise stated).
\end{enumerate}
\end{Lem}

To go further, we need a little bit more lemmas.

The following fact is well known (see, e.g., Ref. \cite{MR2953553}):
\begin{Lem}\label{USAE} Let $ T $ be a symmetric operator in $ \H $. If $ T $ has a unique self-adjoint extension, then $ T $ is essentially self-adjoint.
\end{Lem}
\begin{Lem}\label{a-or-b}
 Let $ T $ be a symmetric operator in $ \H $. If there exists a dense subspace $ V $ such that for any $ \xi \in V $ the initial value problem 
\begin{align}
\frac{d}{dt} \xi (t) = -i T \xi (t) , \q \xi (0 ) = \xi ,
\end{align}
has a solution $ \R \ni t \mapsto \xi (t ) \in D(T) $, then, 
exactly one of the following (a) or (b) holds.
\begin{enumerate}[(a)] 
\item $ T $ has no self-adjoint extension.

\item $ T $ is essentially self-adjoint.
\end{enumerate}
\end{Lem}
\begin{proof} It is sufficient to prove that, if there exists a self-adjoint extension, then $ T $ is essentially self-adjoint. Suppose that $ T $ has a self-adjoint extension $ \hat{T} $. Then, for each $ \eta \in D(\hat{T}) $ and $ \xi \in V $, we have 
\begin{align}
\frac{d}{dt} \left\langle \eta , e^{it\hat{T}} \xi (t) \right\rangle & = \left\langle -i\hat{T} e^{-it\hat{T}} \eta , \xi (t) \right\rangle + \left\langle e^{-it\hat{T}} \eta , -iT \xi (t) \right\rangle .\label{diff1}
\end{align}
Since $\xi(t)$ belongs to $D(T)$, the first term on the right hand side of \eqref{diff1} is equal to $\left\langle -ie^{-it\hat{T}} \eta ,T \xi (t) \right\rangle$. Hence
\begin{align}
\frac{d}{dt} \left\langle \eta , e^{it\hat{T}} \xi (t) \right\rangle =0
\end{align}
for all $t\in \R$.
Thus, we have
\begin{align}
\left\langle \eta , e^{it\hat{T}} \xi (t) \right\rangle = \left\langle \eta, \xi(0) \right\rangle =\left\langle \eta , \xi \right\rangle , \q t \in \R .
\end{align}
Since $ \eta \in D(\hat{T}) $ is arbitrary, $ \xi (t) = e^{-it\hat{T}} \xi \, ( t\in \R ) $ for all $ \xi \in V $. This implies that, if $ T $ has another self-adjoint extension $ \hat{T} ' $, then $ e^{-it\hat{T}} = e^{-it\hat{T}'} \, (t \in \R ) $. Hence, $ \hat{T} = \hat{T} ' $ by Stone's theorem. This means that the self-adjoint extension of $ T $ is unique. Thus, $ T $ is essentially self-adjoint by Lemma \ref{USAE}. 
\end{proof}

The next lemma is related to Stone's theorem and the proof can be found in e.g., Ref. \cite [p. 267]{MR751959}. 
\begin{Lem}\label{stone}
 Let $ T $ be a symmetric operator in the Hilbert space $ \H $. If for any $ \xi \in D(T) $ the initial value problem 
\begin{align}
\frac{d}{dt} \xi (t) = -i \overline{T} \xi (t) , \q \xi (0 ) = \xi ,
\end{align}
has a solution $ \R \ni t \mapsto \xi (t ) \in D(\overline{T}) $, then $ \overline{T} $ is self-adjoint.
\end{Lem}

\begin{proof}[Proof of Therem \ref{abs-main}]
We first prove (i).
From Lemma \ref{time-order}, we find that for all $\xi \in D'$, there is a solution $\xi(t)=e^{-itH_0}\overline{U(t,0)}\xi\in D(H)$ 
of the initial value problem 
generated by $H$,
\begin{align}
\frac{d}{dt} \xi (t) = -i H\xi (t) , \q \xi (0 ) = \xi .
\end{align}
 Since $D'$ is dense, the assertion follows from Lemma \ref{a-or-b}.

Next, we prove (ii).
Let $H'$ be a restriction of $H$ to $D'$. Then $H'$ is a symmetric operator,
and $\overline{H'}=\overline{H}$, because $D'$ is a core of $H$. From Lemma \ref{time-order},
we conclude that for all $\xi \in D'=D(H')$ there is a solution $\xi(t)=e^{-itH_0}\overline{U(t,0)}\xi\in D(H)$ of the initial value problem 
generated by $H$,
\begin{align}
\frac{d}{dt} \xi (t) = -i\overline{H}\xi (t) = -i \overline{H'}\xi (t) , \q \xi (0 ) = \xi 
\end{align}
 From Lemma \ref{stone}, it follows that
$\overline{H'}$ is self-adjoint, but this immediately implies that $H$ is essentially self-adjoint, because
 $\overline{H'}=\overline{H}$.
\end{proof}

\section{Dirac-Maxwell Hamiltonian}
In this section, we introduce the Dirac-Maxwell Hamiltonian and prove its essential
self-adjointness under a suitable condition. 
The Dirac-Maxwell Hamiltonian describes 
a quantum system consisting of a Dirac particle under a potential $V$ and a radiation field minimally interacting
with each other.
We will use the unit system in which the speed of light and $ \hbar $, the Planck constant
devided by $2\pi$, are set to be unity.

Firstly, we consider the Dirac particle sector. Let us denote the mass and the charge of the Dirac particle by $ M>0 $ and $ q \in \Real $, respectively.   The Hilbert space of state vectors for the Dirac particle is taken to be 
\begin{align}
\H _\mathrm{D} := L^2 (\Real_\xx ^3; \C ^4),
\end{align}
the $\C^4$-valued square integrable functions on $\Real_\xx^3 = \{ \xx = (x^1 , x^2 , x^3) \, | \, x^j \in \R , \, j=1,2,3 \} $. 
The vector space $\R^3_\xx$ here represents
the position space of the Dirac particle. We sometimes omit the subscript $\xx$ and just write $\R^3$ instead of $\R^3_\xx$
when no confusion may occur.
The target space $\C^4$ realizes a representation of the four dimensional 
Crifford algebra accompanied by the four dimensional Minkowski vector space. The generators 
$\{\gamma^\mu\}_{\mu=0,1,2,3}$ 
satisfy the anti-commutation relations
\begin{align}
\{\gamma^\mu, \gamma^\nu\}:=\gamma^\mu\gamma^\nu+\gamma^\nu\gamma^\mu=2g^{\mu\nu},\q \mu,\nu=0,1,2,3,
\end{align}
where the Minkowski metric tensor $(g_{\mu\nu})$ is given by
\begin{align}
(g_{\mu\nu})= \left(\begin{matrix}
1 & 0& 0& 0 \\
0 & -1 & 0& 0 \\
0& 0& -1 &0 \\
0&0&0& -1
\end{matrix}\right),
\end{align}
and $g^{\mu\nu}$ denotes the $\mu\nu$-component of the inverse matrix of the above $(g_{\mu\nu})$
(numerically the same as $g_{\mu\nu}$). 
We assume $\gamma^0$ to be Hermitian and $\gamma^j$'s ($j=1,2,3$) be anti-Hermitian.
We use the notations following Dirac :
\begin{align}
\beta&:=\gamma^0 ,\\
\alpha^j &:= \gamma^0 \gamma^j,\q j=1,2,3.
\end{align}
Then, $\alpha^j$'s and $\beta$ satisfy the anti-commutation relations
\begin{align}
& \{ \alpha ^i , \alpha ^j \} = 2 \delta ^{ij} , \q i,j =1,2,3, \\
& \{ \alpha ^j , \beta \} =0 , \q \beta ^2 =1, \q j=1,2,3,
\end{align}
where $ \delta ^{ij} $ is the Kronecker delta.
The momentum operator of the Dirac particle is given by
\begin{align}
\mathbf{p} := (p_1 , p_2 , p_3) := (-i D_1 , -iD_2 , -iD_3)
\end{align}
with $ D_j $ being the generalized partial differential operator on $ L^2 (\Real ^3; \C ^4) $ 
with respect to the variable $ x^j $, the $ j $-th component of $ \mathbf{x} = (x^1 , x^2 , x^3) \in \R ^3 $. 
We write in short
\[  { \boldsymbol \alpha } \cdot \mathbf{p} := \sum _{j=1} ^3 \alpha ^j p_j .\]
The potential is represented by a $ 4\times 4 $ Hermitian matrix-valued function $ V $ on $ \R ^3_\xx$ with each matrix components 
being Borel measurable. Note that the function $V$ naturally defines a multiplication operator acting in $\H_\rm{D}$. We denote it by the same symbol $V$.
The Hamiltonian of the Dirac particle under the influence of this potential $V$ is then given by the Dirac operator
\begin{align}
H_\mathrm{D} (V) := { \boldsymbol \alpha } \cdot \mathbf{p} +M \beta +V 
\end{align}
acting in $\H_{\mathrm{D}}$,
with the domain $ D(H_\mathrm{D} (V)) := H^1 (\Real ^3 ; \C ^4) \cap D(V) $, where $ H^1 (\Real ^3) $
denotes the $ \C ^4 $ valued Sobolev space of order one. 
%We will assume that the potential satisfies the following assumption \ref{ass-V}.
Let $C$ be the conjugation operator in $\mathcal{H}_\rm{D}$ defined by
\[ (Cf)(\xx)=f(\xx)^* ,\q f\in \H_\rm{D},\q \xx\in \R^3,\]
where $*$ means the usual complex conjugation.
By Pauli's lemma \cite{MR1219537}, there is a $4\times 4$ matrix $U$ satisfying
\begin{align}
U^2&=1, \q UC=CU, \\
U^{-1}\alpha^j U &= \overline{\alpha^j}, \q j=1,2,3, \q U^{-1}\beta U = -\overline{\beta},
\end{align}
where for a matrix $A$, $\overline{A}$ denotes its complex-conjugated matrix
and $1$ the identity matrix.
We assume that the potential $V$ satisfies the following conditions :
\begin{Ass}\label{ass-V}
\begin{enumerate}[(I)]
\item Each matrix component of $V$ belongs to
\[ L^2_\rm{loc}(\R^3):=\left\{f:\R^3\to \C\, \Bigg|\, \text{Borel measurable and } \int_{|\xx|\le R}|f(\xx)|^2<\infty \,\text{ for all $R>0$.} \right\}.\]
\item $V$ is Charge-Parity (CP) invariant in the following sense:
\begin{align}
U^{-1}V(\xx)U = V (-\xx)^*, \q \rm{a.e.} \,x\in\R^3.
\end{align}
\item $H_\rm{D}(V)$ is essentially self-adjoint.
\end{enumerate}
\end{Ass}
\noindent Hereafter, we denote the closure of $H_\rm{D}(V)$, which is self-adjoint by Assumption
\ref{ass-V}, by the same symbol.
The important remark is that the Coulomb type potential 
\begin{align}\label{Cou}
V(\xx)=-\frac{Zq^2}{|\xx|} 
\end{align}
satisfies Assumption \ref{ass-V} provided that $Zq^2 <1/2$, or more concretely,
$Z\le 68$ if we put $q=e$, the elementary charge \cite{MR1219537}.
 
Secondly, we introduce the free radiation field Hamiltonian in the Coulomb gauge.
We adopt as the one-photon Hilbert space
\begin{align}
\H _\mathrm{ph} := L^2 (\R ^3 _\kk ; \C ^2) .
\end{align}
The above $ \R ^3 _\kk := \{ \mathbf{k} = (k^1,k^2,k^3) \, | \, k^j \in \R , \, j=1,2,3 \} $ represents the momentum space of photons, and the the target space $\C^2$ represents the
degrees of freedom coming from the polarization of photons. 
We often omit the subscript $ \kk $ in $ \R ^3 _\kk $, and just denote it by $\R^3$, when there is no danger of confusion.
The Hilbert space for the quantized radiation field in the Coulomb gauge is given by 
\begin{align}
\F _\mathrm{ph} :=  \op _{n=0} ^\infty \ot _\rm{s} ^n \H _\rm{ph} = \Big\{ \Psi = \{ \Psi ^{(n)} \} _{n=0} ^\infty \, \Big| \, \Psi ^{(n)} \in \ot _\rm{s} ^n \H _\rm{ph} , \, \, ||\Psi ||^2:=\sum _{n=0} ^\infty \| \Psi ^{(n)} \| ^2 <\infty \Big\} ,
\end{align}
the Boson Fock space over $ \H _\mathrm{ph} $, where $ \ot _\rm{s} ^n $ denotes the $ n $-fold symmetric tensor product with the convention $ \ot _\rm{s} ^0 \H _\rm{ph} := \C $.
Let $ \omega (\kk ) := |\kk | , \, \kk \in \R ^3 $, the energy of a photon with momentum $\kk\in
\Real^3$. The multiplication operator by the $2\times 2$ matrix-valued function 
\begin{align}
\kk\mapsto
\begin{pmatrix}
 \omega (\kk) & 0 \\
 0 &\omega(\kk)
\end{pmatrix}
 \end{align}
acting in $\H_\rm{ph}$ is self-adjoint, and we also denote it by the same symbol $ \omega $. This operator 
$\omega$ is a one-photon Hamiltonian
in $\mathcal{H}_{\rm{ph}}$, and 
the free Hamiltonian (kinetic term) of the quantum radiation field is given by its second quantization
\begin{align}
H_\mathrm{rad} := \mathrm{d}\Gamma _\bb (\omega ) :=\op_{n=0}^\infty \overline{ \Big( \sum _{j=1} ^n I \otimes \dots \otimes  I \otimes \stackrel{j\text{-th}}{\omega} \otimes I \otimes \dots \otimes I \Big) \upharpoonright \hot ^n D(\omega ) } .
\end{align}
The operator $ H_\mathrm{rad} $ is self-adjoint.

Thirdly, we will introduce the point-like quantum radiation field $\mathbf{A}(\xx)$ at $\xx\in \R^3$ with
 an ultraviolet (UV) cut-off in the Coulomb gauge. This is defined
in terms of photon polarization vectors $\{\mathbf{e} _{(r)}\}_{r=1,2}$ and
an ultraviolet cut-off function $\chi\in L^2(\R^3_\xx)$ as follows.
Photon polarization vectors are 
$ \R_\kk ^3 $-valued measurable functions $ \mathbf{e} _{(r)} =(e_{(r)}^1,e_{(r)}^2,e_{(r)}^3)$ ($r=1,2$) on $\R^3_\kk$ 
such that, for all $ \kk \in M_0  := \R ^3 \backslash \{ (0,0, k^3) \, | \, k^3 \in \R \} $ , 
\begin{align}
\ee _{(r)} (\kk ) \cdot \ee _{(r')} (\kk ) = \delta _{rr'} , \q \ee _{(r)} (\kk ) \cdot \kk =0 , \q r,r' =1,2,
\end{align}
where the above $\cdot$ means the usual Euclidean inner product defined by
\[ \ee _{(r)} (\kk ) \cdot \ee _{(r')} (\kk )=\sum_{j=1}^3\e_{(r)}^j (\kk )  e_{(r')}^j (\kk ) ,\q \ee _{(r)} (\kk ) \cdot \kk = \sum_{j=1}^3\e_{(r)}^j (\kk )  k^j,
\q \kk = (k^1,k^2,k^3).\]
 Note that such vector valued functions can be chosen so that they are continuous
 on $M_0$.
An ultraviolet cut-off function $\chi\in L^2(\R^3_\xx)$ is a real valued function
on $ \R ^3_\xx $ satisfying
\begin{align}
\frac{\hat{\chi}}{\sqrt{\omega}} \in L^2 (\R_\kk ^3) ,
\end{align}
where $ \hat{\chi } $ denotes the Fourier transform of $ \chi $. 
Let us denote by $ a(F ) \, (F \in \H _\mathrm{ph})  $ the annihilation operator on $\mathcal{F}_\rm{ph}$, and $\phi(F)$ by
the Segal field operator
\begin{align}
\phi (F) := \frac{\overline{a(F) + a(F) ^*}}{\sqrt{2}}.
\end{align}
It is well known that $\phi(F)$ is self-adjoint. For each $ f \in L^2 (\R ^3) $, we define 
\begin{align}
a^{(1)} (f) := a(f,0) , \q a^{(2)} (f) := a(0,f) .
\end{align}
Then, the point-like quantized radiation field $ \mathbf{A} (\xx ) := (A^1 (\xx ) , A^2 (\xx ) , A^3 (\xx )) $ with the UV cut-off $ \chi $ is given by 
\begin{align}
A^j (\xx ) := \phi (g_\xx^j ) , \q j=1,2,3,
\end{align}
with
\begin{align}
g^j _\xx  (\kk ) := \left( \frac{\hat{\chi} (\kk ) e^j _{(1)} (\kk ) e^{-i \kk \xx }}{\sqrt{\omega (\kk )}} , \frac{\hat{\chi} (\kk ) e^j _{(2)} (\kk ) e^{-i \kk \xx }}{\sqrt{\omega (\kk )}} \right) \in\C^2.
\end{align}

Forthly, we introduce the interaction Hamiltonian and the total Hamiltonian 
in the Hilbert space of state vectors 
for the coupled system, which is
taken to be 
\begin{align}
\F _\mathrm{DM} := \H _\mathrm{D} \ot \F _\mathrm{rad} .
\end{align}
We remark that this Hilbert space can be identified as 
\begin{align}\label{decomposition}
\F _\mathrm{DM} = L^2 (\R_\xx ^3 ; \op ^4 \F _\mathrm{rad}) = \int ^{\oplus} _{\R ^3} d\xx \, \op ^4\F _\mathrm{rad} ,
\end{align}
the Hilbert space of $ \op ^4 \F _\mathrm{rad} $-valued Lebesgue square integrable functions on $ \R_\xx ^3 $ (the constant fibre direct integral with the base space $ (\R ^3 , d\xx ) $ and fibre $ \op ^4 \F _\mathrm{rad} $). We freely use this identification. 
Now, since the mapping $ \xx \mapsto g_\xx^j $ from $ \R ^3 $ to $ \H _\mathrm{ph} $ is strongly continuous, we can define a decomposable self-adjoint operator $A^j$ by
\begin{align}
A^j := \int ^\oplus _{\R ^3} d\xx \, A^j (\xx ) ,\q j=1,2,3,
\end{align}
acting in $ \int ^\oplus _{\R ^3} d\xx \, \op^4 \F _\mathrm{rad}  $.
We have now arrived at the position to define the minimal interaction
Hamiltonian $H_1$, between the Dirac particle and the quantized radiation field with the UV cutoff $ \chi $. It is given by
\begin{align}
H_1 := -q { \boldsymbol \alpha } \cdot \mathbf{A}=-q \sum_{j=1}^3\alpha^j A^j.
\end{align}
The total Hamiltonian of the coupled system is then given by
\begin{align}
H_\mathrm{DM} (V) &:= H_0  + H_1, \\
 H_0& :=H_\mathrm{D} (V) + H_\mathrm{rad}.
\end{align}
This is called the \textit{Dirac-Maxwell Hamiltonian}. The essential self-adjointness of $ H_\mathrm{DM} (V) $ is discussed in \cite{MR1765584}. However, when the potential $ V $ is of the Coulomb
type \eqref{Cou}, the essential self-adjointness of $ H_\mathrm{DM} (V) $ remains to be proved.

The rest of the present paper is devoted to prove
\begin{Thm}\label{main}
Suppose that the potential $ V $ satisfies Assumption \ref{ass-V}, and that the Fourier transformation
of the UV cut-off, which we denote by $\hat\chi$, is real-valued. Then, $H_\rm{DM}(V)$ is essentially self-adjoint.
\end{Thm}
\noindent We emphasize here again 
that Theorem \ref{main} certainly covers the Coulomb potential case \eqref{Cou}, if $Zq^2<1/2$.

\section{Proof of Theorem \ref{main}}
First, we recall the important result obtained in Ref. \cite{MR1765584}.
\begin{Lem}\label{exi-of-sa-ext}
Suppose that Assumption \ref{ass-V} is valid and $\hat\chi$ is real-valued. Then,
$H_\rm{DM}(V)$ has a self-adjoint extension.
\end{Lem}

\begin{proof}
See Ref. \cite{MR1765584}, Theorem 1.2.
\end{proof}

Let $N_\bb$ be the photon number operator which is defined by
\begin{align}
N_\bb := 1\ot \mathrm{d}\Gamma_\bb(1), 
\end{align}
acting in $ \H _\rm{D} \otimes \F _\rm{rad} $.
Note that the operator $ N_\bb $ can be identified with the following decomposable operator 
in the sense of \eqref{decomposition}:
\begin{align}
N_\bb = \int _{\R ^3} ^\oplus d\xx \, \mathrm{d} \Gamma _\bb (1) .
\end{align}

\begin{Lem}\label{time-od-lem}
The Dirac Maxwell Hamiltonian $H_\rm{DM}(V)$ fulfills Assumption \ref{ass1},
where the photon number operator $N_\bb$ plays a role of $A$ in it. Namely,   
\begin{enumerate}[(i)]
\item $N_\bb$ is self-adjoint and non-negative.
\item $N_\bb$ and $H_0$ are strongly commuting.
\item $H_1$ is $N_\bb^{1/2}$-bounded.
\item If $\Psi\in E_{N_\bb}([0,L])$, then $H_1\Psi \in E_{N_\bb}([0,L+1])$.
\end{enumerate}
\end{Lem}
\begin{proof}
The assertion (i) and (ii) are well known. 

We prove (iii). It is well known \cite{AraiFock} that for all $\Psi\in D(\mathrm{d}\Gamma_\bb(1)^{1/2})$ and $F\in\H_\rm{ph}$,
$\Psi$ belongs to $D(a(F)\cap D(a(F)^*)$ and the estimates
\begin{align}
|| a(F)\Psi || \le \| F \|\, ||\mathrm{d}\Gamma_\bb(1) ^{1/2}\Psi || ,\q ||a(F)^*\Psi || \le ||F|| \,||\mathrm{d}\Gamma_\bb(1) ^{1/2} \Psi || + ||F||\,||\Psi||
\end{align}
are valid. Therefore, we find that, for all $\Psi\in D(\mathrm{d}\Gamma_\bb(1)^{1/2})$, $\xx\in\R^3$ and $j=1,2,3$,
$\Psi$ is in $D(A^j(\xx))$ and that
\begin{align}
|| A^j (\xx) \Psi || \le \sqrt{2} ||g_\xx^j || \,|| \mathrm{d}\Gamma_\bb(1) ^{1/2} \Psi || +\frac{1}{\sqrt{2}} ||g_\xx^j|| \,||\Psi ||. 
\end{align}
Take arbitrary $\Psi\in D(N_\bb^{1/2})$ as a subspace of $\F_\rm{DM}$. Then, for almost
every $\xx\in\Real^3$ and $j=1,2,3$, we have $\Psi(\xx) \in D(A^j(\xx))$ and
\begin{align}
|| A^j (\xx) \Psi (\xx)|| \le \sqrt{2} ||g_\xx^j || \,|| \mathrm{d}\Gamma_\bb(1)^{1/2}\Psi(\xx) || +\frac{1}{\sqrt{2}} ||g_\xx^j|| \,||\Psi(\xx) ||. 
\end{align}
From the elementary inequality
\[ (a+b)^2 \le 2a^2+2b^2 ,\]
and the fact that $||g_\xx^j || = ||g_\mathbf{0}^j ||$,
it follows that
\begin{align}\label{at-x}
|| A^j (\xx) \Psi (\xx)||^2 \le 4 ||g_\mathbf{0}^j ||^2 \,|| \mathrm{d}\Gamma_\bb(1) ^{1/2} \Psi(\xx) ||^2 + ||g_\mathbf{0}^j||^2 \,||\Psi(\xx) ||^2. 
\end{align}
By integrating both sides of \eqref{at-x} with respect to $\xx$ on $\R^3$, one obtains
$\Psi\in D(A^j)$ and
\begin{align}
|| A^j \Psi ||^2 &= \int_{\R^3} d\xx \,|| A^j (\xx) \Psi (\xx)||^2 \no\\
	&\le 4 ||g_\mathbf{0}^j ||^2 \,\int_{\R^3} d\xx \,|| \mathrm{d}\Gamma_\bb(1) ^{1/2}\Psi(\xx) ||^2 
	+ ||g_\mathbf{0}^j||^2 \,\int_{\R^3} d\xx \,||\Psi(\xx) ||^2 \no \\
	&=4 || g_\mathbf{0}^j ||^2 \,|| N_\bb^{1/2}\Psi ||^2 +  || g_\mathbf{0}^j ||^2 \,||\Psi ||^2.
\end{align}
Hence, we obtain
\begin{align}\label{est-Aj}
|| A^j \Psi || \le2 || g_\mathbf{0}^j || \,|| N_\bb^{1/2}\Psi || +  || g_\mathbf{0}^j ||\,||\Psi ||.
\end{align}
Thus, one obtains from \eqref{est-Aj} the estimate for all $\Psi\in D(N_\bb^{1/2})$,
\begin{align}
|| H_1 \Psi || &= |q| \sum_{j=1}^3 || \alpha^j A^j \Psi ||\no\\
	&\le |q| \sum_{j=1}^3 || \alpha^j || \,|| A^j \Psi ||\no\\
	&\le 2|q| \sum_{j=1}^3 || \alpha^j ||\,|| g_\mathbf{0}^j || \,|| N_\bb^{1/2}\Psi || 
	+ |q| \sum_{j=1}^3 || \alpha^j ||\, || g_\mathbf{0}^j ||\,||\Psi ||,
\end{align}
which proves (iii).

We prove (iv). Let us introduce the closed subspace $\F_N \subset \F _\rm{DM}$ for $N\in\Natural\cup\{0\}$ by
\begin{align}
\F_N:= \H_\rm{D}\ot\left(\op_{n=0}^N \ot_{s}^n \H_\rm{ph}\right) .
\end{align}
As is well known, the spectrum of $N_\bb$ is equal to the discrete set $\Natural\cup\{0\}$,
and the eigenspace belonging to an eigenvalue $n\in\Natural\cup\{0\}$ is 
$\H_\rm{D}\ot\left(\ot_{s}^n \H_\rm{ph}\right)$. Hence, for $L\ge 0$,
$E_{N_\bb}([0,L])$ is the orthogonal projection onto $\mathcal{F}_{[L]}$ 
with $[L]$ denoting the integer satisfying $L-1< [L]\le L$, and we have
\begin{align}
R(E_{N_\bb}([0,L])) = \mathcal{F}_{[L]}.
\end{align}
Since the interaction Hamiltonian $H_1$ creates at most one photon,
$H_1$ maps $\F_N$ into $\F_{N+1}$, and
we find that $\Psi \in R(E_{N_\bb}([0,L]))$ implies $\Psi\in R(E_{N_\bb}([0,[L]+1]))=R(E_{N_\bb}([0,L+1]))$.
This proves (iv).
\end{proof}
%Here, we remark that the algebraic infinite direct sum
%\[ \F_{\bb,0}:=\hop_{n=0}^\infty \left(\H_\rm{D}\ot\left(\ot_{s}^n \H_\rm{ph}\right) \right) \]
%is equal to 
%\[ \bigcup_{L\ge 0} R(E_{N_\bb}([0,L])) \]
%and the subspace $D$ in the abstract setup is now equal to $\F_{\bb,0}$. 

\begin{proof}[Proof of Theorem \ref{main}]
From Lemma \ref{time-od-lem} and Theorem \ref{abs-main} (\ref{main1}),
we find that if there exists at least one self-adjoint extension of $H_\rm{DM}(V)$,
then $H_\rm{DM}(V)$ is essentially self-adjoint. But from Lemma \ref{exi-of-sa-ext}, 
$H_\rm{DM}(V)$ indeed has a self-adjoint extension. This completes the proof.
\end{proof}

We remark that our proof presented here is also applicable to similar 
particle-field Hamiltonians. For instance, we can prove that the Dirac-Klein-Gordon Hamiltonian $H_\mathrm{DKG}(V)$
acting in the Hilbert space $L^2(\R^3;\C^4)\otimes \mathcal{F}_\bb(L^2(\Real^3))$, given by 
\begin{align}
H_\mathrm{DKG}(V)=H_\mathrm{D}(V)+1\otimes d\Gamma_\bb(\omega) 
+\lambda \int^{\oplus}_{\R^3}d\xx\,\beta\, \phi(\chi_j(\xx)),
\end{align}
with
\begin{align}
\chi_j(\xx)(\kk):=\frac{\hat{\chi}(\kk)}{\sqrt{\omega(\kk)}}e^{-i\kk\xx}, \quad \chi\in\ L^2(\Real)
\end{align}
and $\lambda\in\Real$, is essentially self-adjoint as long as the above assumptions 
are satisfied. This Hamiltonian describes a quantum system of a Dirac particle under the potential $V$
 interacting with a neutral scalar field.
\section*{Acknowledgement}
The authors are grateful to Professor Asao Arai for valuable comments and discussions.
They also thank Assistant Professor Toshimitu Takaesu for his comments.

\bibliographystyle{plain}

\bibliography{myref-1}

\end{document}